\documentclass[oneside,english,onecolumn]{amsart}
\usepackage[latin9]{inputenc}
\usepackage{geometry}
\usepackage{amsthm}
\usepackage{graphicx}
\usepackage{amsmath}
\usepackage{amssymb}

\makeatletter
\numberwithin{equation}{section}
\numberwithin{figure}{section}
\theoremstyle{plain}
\newtheorem{thm}{\protect\theoremname}
\newtheorem*{properties*}{Properties}
  \theoremstyle{plain}
  \newtheorem{lem}[thm]{\protect\lemmaname}
  \theoremstyle{definition}
  \newtheorem{defn}[thm]{\protect\definitionname}
  \theoremstyle{plain}
  
  \theoremstyle{plain}
  
  \theoremstyle{plain}

  \theoremstyle{plain}
  
\newtheorem{rem}{\protect\remarkname}
  \theoremstyle{plain}

  \theoremstyle{plain}
  
\@ifundefined{showcaptionsetup}{}{%
 \PassOptionsToPackage{caption=false}{subfig}}
\usepackage{subfig}
\makeatother

\usepackage{babel}
\providecommand{\propsname}{Properties}
  \providecommand{\corollaryname}{Corollary}
  \providecommand{\conjecturename}{Conjecture}
  \providecommand{\definitionname}{Definition}
  \providecommand{\factname}{Fact}
  \providecommand{\lemmaname}{Lemma}
\providecommand{\theoremname}{Theorem}
\providecommand{\remarkname}{Remark}

\begin{document}

\title{Symmetric Laplacians, Quantum Density Matrices and their Von-Neumann Entropy}

\author{David E. Simmons, Justin P. Coon, and Animesh Datta}
\begin{abstract}
We show that the (normalized) symmetric Laplacian of a simple graph can be obtained from the partial trace over a pure bipartite quantum state that resides in a bipartite Hilbert space (one part corresponding to the vertices, the other corresponding to the edges). This suggests an interpretation of the symmetric Laplacian's Von Neumann entropy as a measure of bipartite entanglement present between the two parts of the state. We then study extreme values for a connected graph's generalized  R\'enyi-$p$ entropy. Specifically, we show that
\begin{enumerate}\item the complete graph achieves maximum entropy, 
\item the $2$-regular graph:\begin{enumerate}
\item achieves minimum R\'enyi-$2$ entropy among all $k$-regular graphs, 
\item is within $\log 4/3$ of the minimum R\'enyi-$2$ entropy and $\log4\sqrt{2}/3$ of the minimum Von Neumann entropy among all connected graphs,
\item achieves a Von Neumann entropy less than the star graph.
\end{enumerate}
\end{enumerate}
Point $(2)$ contrasts sharply with similar work applied to (normalized) combinatorial Laplacians, where it has been shown that the star graph almost always achieves minimum Von Neumann entropy. In this work we find that the star graph achieves maximum entropy in the limit as the number of vertices grows without bound.

\end{abstract}
 
\maketitle

\section{Introduction\label{sec:intro}}

The mathematical theory of quantum mechanics allows us to view quantum states of finite-dimensional systems as Hermitian, positive semi-definite matrices with unit trace \cite{von1955mathematical}. It is also known that the combinatorial Laplacian of a graph is a symmetric positive semi-definite matrix. Thus, normalizing this matrix by its trace allows us to view the graph as a quantum state, \cite{braunstein2006laplacian}. A natural next step is to study the information content of the graph by considering the Von Neumann entropy of the graph's corresponding quantum state, \cite{passerini2008neumann}. In that work, it was noted that the Von Neumann entropy may be interpreted as a measure of network regularity. In \cite{anand2009entropy}, it was then shown that for scale free networks the Von Neumann entropy of a graph is linearly related to the Shannon entropy of the graph's ensemble. Correlations between these entropies were observed when the graph's degree distribution displayed heterogeneity in \cite{anand2011shannon}. It was not until the work of \cite{de2016interpreting} that a well defined operational interpretation of the graph's Von Neumann entropy was established. Specifically, in that work it was demonstrated that the Von Neumann entropy should be interpreted as the number of entangled bits that are represented by the graph's quantum state. Further extensions were made in \cite{dairyko2016note}, where it was shown that, as the number of vertices grows, almost all connected graphs have a Von Neumann entropy that is greater than the star graph.

In this work, 
we show that the graph's \emph{symmetric} (i.e., not combinatorial) Laplacian can be obtained from the partial trace of a particular bipartite pure state. With this, we immediately conclude that the Von Neumann entropy should be interpreted as a measure of bipartite entanglement present within this pure state. We then proceed to study the entropic properties of this state: we show that the complete graph achieves maximum Von Neumann entropy, while the $2$-regular graph is within $\log4\sqrt{2}/3$ of the smallest Von Neumann entropic values among all connected graphs. This contrasts with work in \cite{dairyko2016note}, where it was conjectured that the \emph{star graph} minimizes the Von Neumann entropy. Here we show that the star graph (asymptotically) maximizes the Von Neumann entropy. We also study the generalized  R\'enyi-$p$ entropy, showing that the $2$-regular graph minimizes the R\'enyi-$2$ entropy among all $k$-regular graphs. To the best of the authors' knowledge this is the first time such a study has been performed.

 The remainder of this paper is organized as follows. Section \ref{sec:graphsQStates} introduces basic graph theoretic concepts and proceeds to establish a relationship between the partial trace of a particular bipartite pure state and the symmetric Laplacian. Section \ref{sec:extremal} then establishes the extremal properties of the entropy for connected graphs. Finally, section~\ref{sec:conc} concludes the paper.

 \subsection{Notation}

We use $\mathcal{K}_n$ to denote the complete graph on $n$ vertices, $\mathcal{K}_{n-k,k}$ to denote the complete bipartite (one part with $n-k$ vertices, the other with $k$ vertices) graph on $n$ vertices, and $\mathcal{R}_{k,n}$ to denote the $k$-regular graph on $n$ vertices.

\section{Symmetric Laplacians and Quantum states\label{sec:graphsQStates}}

\subsection{Basic Graph Theory\label{sec:intrographs}}

Consider an undirected graph $G=\left(V,E\right)$, which consists of a vertex set 
\begin{equation}
V=\left\{ v_{1},\dots,v_{n}\right\} , 
\end{equation}
with $\left|V\right|=n$, and an edge set 
\begin{equation}
E=\left\{ e_{1},\dots,e_{m}\right\} ,
\end{equation}
 with $\left|E\right|=m.$ We assume that the element of the edge set (here we assume it is the $k$th element) connecting vertex pair $\left(i,j\right)$ can be written as \begin{equation}e_{k}:=v_{i,j}.\label{eq:edge_vertexeq}\end{equation} The degree matrix of the graph is given by 
\begin{equation}
\Delta =\mathrm{diag}\left\{ d_{1},\dots,d_{n}\right\} ,
\end{equation}
where $d_{i}$ is the degree of $v_{i}$ (i.e., the number of vertices connected to $v_{i}$). The $\left(i,j\right)$th element of a graph's adjacency matrix is given by
\begin{equation}
\left[A \right]_{i,j}=1
\end{equation}
if vertex pair $\left(i,j\right)$ is connected by an edge. The (combinatorial) Laplacian of the graph is then defined to be  \cite{newman2000laplacian}
\begin{equation}
L =\Delta -A ,\label{eq:standLapl}
\end{equation}
while the positive Laplacian is defined to be 
\begin{equation}
L^+ =\Delta +A  .
\end{equation}
The symmetric Laplacian is then given by \cite{chung1997spectral}
\begin{equation}
\mathcal{L} =\Delta ^{-1/2}L \Delta ^{-1/2}  ,\label{eq:symmetricLap}
\end{equation}
while the positive symmetric Laplacian is given by
\begin{equation}
\mathcal{L}^+ =\Delta ^{-1/2}L^+ \Delta ^{-1/2}.\label{eq:symmetricLap2}
\end{equation}
We can also write the $(i,j)$th element of $\mathcal{L}$ as
\begin{equation}
\left[\mathcal{L}\right]_{i,j}=\left\{ \begin{array}{cc}
1 & \mathrm{if}\;i=j\\
-\frac{1}{\sqrt{d_{i}d_{j}}} & \mathrm{if}\;i\;\mathrm{and}\;j\mathrm{\;are\;neighbors}\\
0 & \mathrm{otherwise}
\end{array}.\right.\label{eq:ijthLaplacian}
\end{equation}

\subsection{Symmetric Laplacians as Quantum States}

To associate the symmetric graph Laplacian with a quantum state, we take an approach inspired by \cite{de2016interpreting}. Firstly, if we impose an orientation on the graph, the symmetric Laplacian can be decomposed as follows:
\begin{equation}
\mathcal{L} =S S^\dagger,
\end{equation}
where $S := \Delta ^{-1/2} M$ and $M$ is defined to be 
\[
\left[M\right]_{v,e}=\left\{ \begin{array}{cc}
1, & \mathrm{if}\;v\;\mathrm{is\;a\;source}\\
-1, & \mathrm{if}\;v\;\mathrm{is\;a\;sink}\\
0, & \mathrm{otherwise}.
\end{array}\right.
\]
Note, $MM^\dagger = L$, \cite{de2016interpreting}. Interestingly, however, the symmetric Laplacian of a graph is independent of the orientation
of $S$ (i.e., swapping any sink with its source does
not affect the Laplacian) \cite{chung1997spectral}. 

We can also write $S$ as a sum of outer products.
To do this, we first put the vertex set into one-to-one correspondence
with elements from a set of $n$ orthonormal basis vectors
\begin{equation}
\mathcal{V}=\left\{ \mathbf{v}_{1},\dots,\mathbf{v}_{n}\right\} ,\label{eq:vertexset}
\end{equation}
and the edge set into one-to-one correspondence with elements from
a set of $m$ orthonormal basis vectors
\[
\mathcal{E}=\left\{ \mathbf{e}_{1},\dots,\mathbf{e}_{m}\right\} .
\]
In a similar manner to \eqref{eq:edge_vertexeq}, we assume that the
element of $\mathcal{E}$ (here, the $k$th element)
corresponding to the edge connecting vertex pair $v_{i}$ and $v_{j}$
can be written as 
\begin{equation}
\mathbf{e}_{k}:=\mathbf{v}_{i,j}.\label{eq:edgevertexcorr}
\end{equation}
We can then express $S$ as 
\begin{equation}
S=\sum_{\mathbf{v}_{i,j}\in\mathcal{E}}f_{\mathbf{v}_{i,j}}\left(\mathbf{v}_{i}\right)\left(\frac{\mathbf{v}_{i}}{\sqrt{d_i}}-\frac{\mathbf{v}_{j}}{\sqrt{d_j}}\right)\mathbf{v}_{i,j}^{\dagger},\label{eq:standardincidencematrix}
\end{equation}
where 
\[
f_{\mathbf{v}_{i,j}}\left(\mathbf{v}_{i}\right)=\left\{ \begin{array}{cc}
+1 & \mathrm{if}\;\mathbf{v}_{i}\mathrm{\;is\;a\;source\;for\;edge}\;\mathbf{v}_{i,j}\\
-1 & \mathrm{if}\;\mathbf{v}_{i}\mathrm{\;is\;a\;sink\;for\;edge}\;\mathbf{v}_{i,j}.
\end{array}\right.
\]

To aid our exposition, in what follows we would like to construct
$S$ without having to specify an orientation. This can be done
by doubling the edge space, where for each vertex pair $\left(v_{i},v_{j}\right)$
connected by a directed edge, we introduce the arcs $v_{i}\to v_{j}$
and $v_{j}\to v_{i}$. As an example, Fig. \ref{fig:directedGraph}
shows a directed graph $G$ with vertex set $\left\{ 1,2,3\right\} $
and arc set $\left\{ 1\to2,1\to3\right\} $. The corresponding undirected
graph is shown in Fig. \ref{fig:undirectedGraph}, which is constructed
from the vertex set $\left\{ 1,2,3\right\} $ with arc set $\left\{ 1\to2,2\to1,1\to3,3\to1\right\} $.
As before, the vertex set can be put into one-to-one correspondence
with $\mathcal{V}$ (see (\ref{eq:vertexset})), while the (doubled)
edge set can be put into one-to-one correspondence with elements from
a set of $2m$ orthonormal basis vectors
\[
\mathcal{E}=\left\{ \mathbf{e}_{1},\dots,\mathbf{e}_{2m}\right\} .
\]
In a similar manner to \eqref{eq:edge_vertexeq} and \eqref{eq:edgevertexcorr},
we assume that the element of $\mathcal{E}$ (here,
the $k$th element) corresponding to the arc $i\to j$ can be written
as 
\[
\mathbf{e}_{k}:=\mathbf{v}_{i,j}.
\]
Note, in \eqref{eq:edgevertexcorr}
we did not draw a distinction between $\mathbf{v}_{i,j}$ and $\mathbf{v}_{j,i}$,
and we had $\left|\mathbf{v}_{i,j}^{\dagger}\mathbf{v}_{j,i}\right|=1$.
However, now $\mathbf{v}_{i,j}$ and $\mathbf{v}_{j,i}$ correspond
to unique (orthogonal) arcs and we have $\mathbf{v}_{i,j}^{\dagger}\mathbf{v}_{j,i}=0.$
From here on we can interpret $\mathbf{v}_{i,j}$ in the following way
\[
\mathbf{v}_{i,j}=\mathbf{v}_{i}\otimes\mathbf{v}_{j}
\]
whenever an edge exists between vertices $i$ and $j$; otherwise, 
\[
\mathbf{v}_{i,j}=\mathbf{0}\otimes\mathbf{0}.
\]

We will now construct the incidence matrix of the graph without having to specify an orientation.
We can partition $\mathcal{E}$ into two disjoint subsets $\mathcal{E}_{1}$
and $\mathcal{E}_{2}$, where $\mathcal{E}_{1}$ canonically characterizes
one orientation of the graph while ${\mathcal{E}}_{2}$ characterizes
the complimentary orientation; i.e., for each $\mathbf{v}_{i,j}\in\mathcal{E}_{1}$ there
is a corresponding $\mathbf{v}_{j,i}\in\mathcal{E}_{2}$. We can then
write $\bar{S}$ (the corresponding $S$ matrix for the graph after the edge doubling) as
the following sum of outer products
\begin{eqnarray}
\bar{S} & = & \sum_{\mathbf{v}_{i,j}\in\mathcal{E}}f_{\mathbf{v}_{i,j}}\left(\mathbf{v}_{i}\right)\left(\frac{\mathbf{v}_{i}}{\sqrt{d_i}}-\frac{\mathbf{v}_{j}}{\sqrt{d_j}}\right)\mathbf{v}_{i,j}^{\dagger}\nonumber \\
 & = & \sum_{\mathbf{v}_{i,j}\in\mathcal{E}_{1}}f_{\mathbf{v}_{i,j}}\left(\mathbf{v}_{i}\right)\left(\frac{\mathbf{v}_{i}}{\sqrt{d_i}}-\frac{\mathbf{v}_{j}}{\sqrt{d_j}}\right)\mathbf{v}_{i,j}^{\dagger}+\sum_{\mathbf{v}_{j,i}\in\mathcal{E}_{2}}f_{\mathbf{v}_{j,i}}\left(\mathbf{v}_{j}\right)\left(\frac{\mathbf{v}_{j}}{\sqrt{d_i}}-\frac{\mathbf{v}_{i}}{\sqrt{d_j}}\right)\mathbf{v}_{j,i}^{\dagger}\nonumber \\
 & = & \sum_{\mathbf{v}_{i,j}\in\mathcal{E}_{1}}f_{\mathbf{v}_{i,j}}\left(\mathbf{v}_{i}\right)\left(\frac{\mathbf{v}_{i}}{\sqrt{d_i}}-\frac{\mathbf{v}_{j}}{\sqrt{d_j}}\right)\mathbf{v}_{i,j}^{\dagger}-\sum_{\mathbf{v}_{j,i}\in\mathcal{E}_{2}}f_{\mathbf{v}_{j,i}}\left(\mathbf{v}_{j}\right)\left(\frac{\mathbf{v}_{i}}{\sqrt{d_i}}-\frac{\mathbf{v}_{j}}{\sqrt{d_j}}\right)\mathbf{v}_{j,i}^{\dagger}\nonumber \\
 & = & \sum_{\mathbf{v}_{i,j}\in\mathcal{E}_{1}}\left(\frac{\mathbf{v}_{i}}{\sqrt{d_i}}-\frac{\mathbf{v}_{j}}{\sqrt{d_j}}\right)\mathbf{v}_{i,j}^{\dagger}-\sum_{\mathbf{v}_{j,i}\in\mathcal{E}_{2}}\left(\frac{\mathbf{v}_{i}}{\sqrt{d_i}}-\frac{\mathbf{v}_{j}}{\sqrt{d_j}}\right)\mathbf{v}_{j,i}^{\dagger}\nonumber \\
 & = & \sum_{\mathbf{v}_{i,j}\in\mathcal{E}_{1}}\left(\frac{\mathbf{v}_{i}}{\sqrt{d_i}}-\frac{\mathbf{v}_{j}}{\sqrt{d_j}}\right)\left(\mathbf{v}_{i,j}-\mathbf{v}_{j,i}\right)^{\dagger}\nonumber \\
 & = & \sum_{\mathbf{v}_{i,j}\in\mathcal{E}_{2}}\left(\frac{\mathbf{v}_{i}}{\sqrt{d_i}}-\frac{\mathbf{v}_{j}}{\sqrt{d_j}}\right)\left(\mathbf{v}_{i,j}-\mathbf{v}_{j,i}\right)^{\dagger},\label{eq:undirectedincidence}
\end{eqnarray}
where the first line follows from (\ref{eq:standardincidencematrix}),
the second line follows from the partition of $\mathcal{E}$ into
two disjoint subsets $\mathcal{E}_{1}$ and $\mathcal{E}_{2}$, the
third and fourth lines follow from the fact that $f_{\mathbf{v}_{i,j}}\left(\mathbf{v}_{i}\right)=f_{\mathbf{v}_{j,i}}\left(\mathbf{v}_{j}\right)=1$,
and the last two lines follow from the symmetry of the partition of
$\mathcal{E}$ into $\mathcal{E}_{1}$ and $\mathcal{E}_{2}$. We
then have
\begin{eqnarray}
\bar{S}\bar{S}^{\dagger} & = & \left(\sum_{\mathbf{v}_{i,j}\in\mathcal{E}_{1}}\left(\frac{\mathbf{v}_{i}}{\sqrt{d_i}}-\frac{\mathbf{v}_{j}}{\sqrt{d_j}}\right)\left(\mathbf{v}_{i,j}-\mathbf{v}_{j,i}\right)^{\dagger}\right)\left(\sum_{\mathbf{v}_{i,j}\in\mathcal{E}_{1}}\left(\frac{\mathbf{v}_{i}}{\sqrt{d_i}}-\frac{\mathbf{v}_{j}}{\sqrt{d_j}}\right)\left(\mathbf{v}_{i,j}-\mathbf{v}_{j,i}\right)^{\dagger}\right)^{\dagger}\nonumber\\
 & = & \left(\sum_{\mathbf{v}_{i,j}\in\mathcal{E}_{1}}\left(\frac{\mathbf{v}_{i}}{\sqrt{d_i}}-\frac{\mathbf{v}_{j}}{\sqrt{d_j}}\right)\left(\mathbf{v}_{i,j}-\mathbf{v}_{j,i}\right)^{\dagger}\right)\left(\sum_{\mathbf{v}_{i,j}\in\mathcal{E}_{1}}\left(\mathbf{v}_{i,j}-\mathbf{v}_{j,i}\right)\left(\frac{\mathbf{v}_{i}}{\sqrt{d_i}}-\frac{\mathbf{v}_{j}}{\sqrt{d_j}}\right)^{\dagger}\right)\nonumber\\
 & = & \sum_{\mathbf{v}_{i,j}\in\mathcal{E}_{1}}\left(\frac{\mathbf{v}_{i}}{\sqrt{d_i}}-\frac{\mathbf{v}_{j}}{\sqrt{d_j}}\right)\left(\mathbf{v}_{i,j}-\mathbf{v}_{j,i}\right)^{\dagger}\left(\mathbf{v}_{i,j}-\mathbf{v}_{j,i}\right)\left(\frac{\mathbf{v}_{i}}{\sqrt{d_i}}-\frac{\mathbf{v}_{j}}{\sqrt{d_j}}\right)^{\dagger}\nonumber\\
 & = & 2\sum_{\mathbf{v}_{i,j}\in\mathcal{E}_{1}}\left(\frac{\mathbf{v}_{i}}{\sqrt{d_i}}-\frac{\mathbf{v}_{j}}{\sqrt{d_j}}\right)\left(\frac{\mathbf{v}_{i}}{\sqrt{d_i}}-\frac{\mathbf{v}_{j}}{\sqrt{d_j}}\right)^{\dagger}\nonumber\\
 & = & 2SS^{\dagger}\nonumber\\
 & = & 2\mathcal{L} .\label{eq:SbarSbarlap}
\end{eqnarray}
The fifth line follows from \eqref{eq:standardincidencematrix}, where the factor of $2$ is a result of the edge doubling procedure that 
was performed; while the final line is a standard result (see \cite{chung1997spectral}). 
\begin{figure}
\centering{}\subfloat[\label{fig:directedGraph}An example of a digraph.]{\centering{}\includegraphics[scale=0.5]{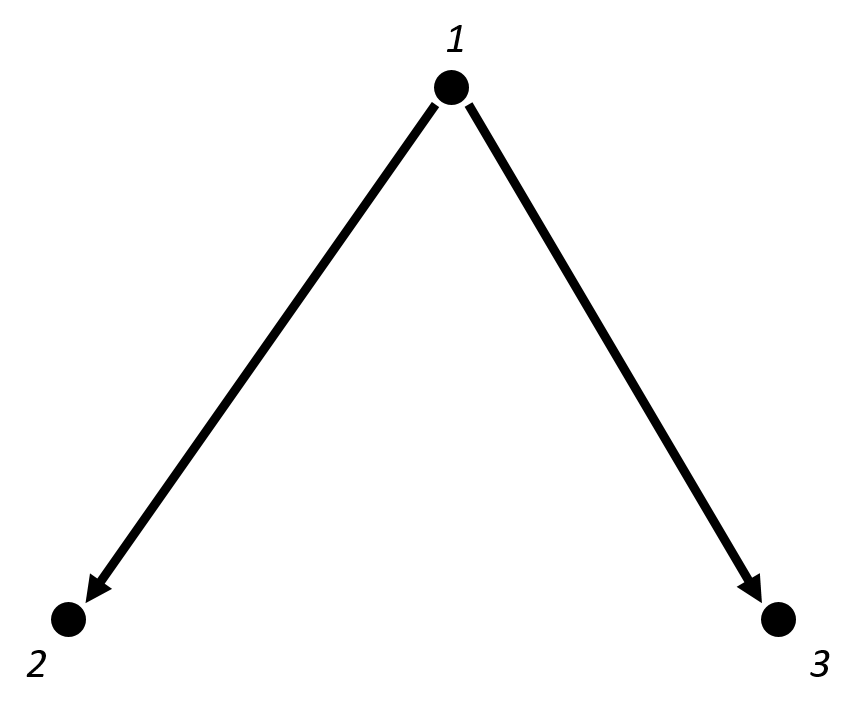}}\subfloat[\label{fig:undirectedGraph}An example of the undirected version of
the digraph presented in Fig. \ref{fig:directedGraph}.]{\centering{}\includegraphics[scale=0.5]{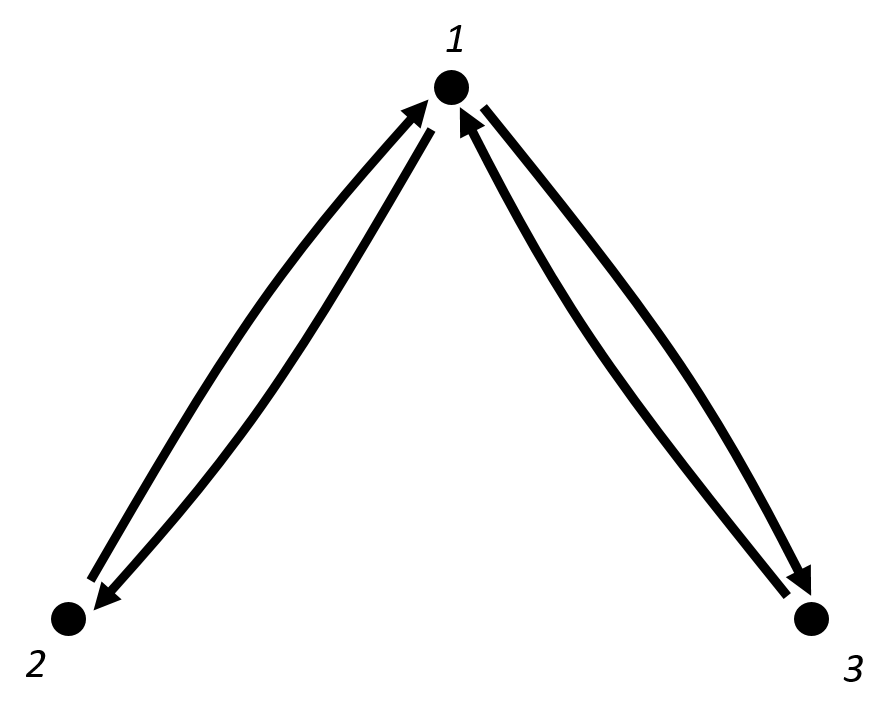}}\caption{}
\end{figure}

Because of the relationship between outer products and tensor products,
(\ref{eq:undirectedincidence}) highlights implicitly that we can
associate the element 
\begin{equation}
\psi_{G}:=\frac{1}{\sqrt{2}}\sum_{\mathbf{v}_{i,j}\in\mathcal{E}_{1}}\left(\frac{\mathbf{v}_{i}}{\sqrt{d_i}}-\frac{\mathbf{v}_{j}}{\sqrt{d_j}}\right)\otimes\left(\mathbf{v}_{i,j}-\mathbf{v}_{j,i}\right)\label{eq:qstate_looped}
\end{equation}
 from the tensor product Hilbert space $\mathcal{H}_{\mathcal{V}}\otimes\mathcal{H}_{\mathcal{E}}$ ($\mathcal{H}_{\mathcal{V}}$ corresponding to the vertex set and $\mathcal{H}_{\mathcal{E}}$ corresponding to the edge set),
where 
\[
\mathcal{H}_{\mathcal{E}}:=\mathcal{H}_{\mathcal{V}}\otimes\mathcal{H}_{\mathcal{V}},
\]
 with $\bar{S}$, so that $\psi_{G}$ is proportional to the
vectorization of $\bar{S}$. Finally, we reach the following lemma
\begin{lem}
\label{lem:SymLaplacianQstate}With $\psi_{G}$ given by (\ref{eq:qstate_looped}),
and $\mathcal{L}  $ and $\mathcal{L}^{+} $ given in
\eqref{eq:symmetricLap} and \eqref{eq:symmetricLap2}, we have
\begin{equation}
\mathrm{Tr}_{\mathcal{E}}\left\{ \psi_{G}\psi_{G}^{\dagger}\right\} =\mathcal{L} .
\end{equation}
and 
\begin{equation}
\mathrm{Tr}_{\mathcal{V}}\left\{ \psi_{G}\psi_{G}^{\dagger}\right\} =\mathcal{L}^{+} .
\end{equation}
\end{lem}
\begin{proof}
We present the proof for $\mathrm{Tr}_{\mathcal{E}}\left\{ \psi_{G}\psi_{G}^{\dagger}\right\} $.
Similar calculations hold for $\mathrm{Tr}_{\mathcal{V}}\left\{ \psi_{G}\psi_{G}^{\dagger}\right\} $.
We have 
\begin{eqnarray}
 \mathrm{Tr}_{\mathcal{E}}\left\{ \psi_{G}\psi_{G}^{\dagger}\right\} 
 & = & \frac{1}{2}\mathrm{Tr}_{\mathcal{E}}\left\{ \left(\sum_{\mathbf{v}_{i,j}\in\mathcal{E}_{1}}\left(\frac{\mathbf{v}_{i}}{\sqrt{d_i}}-\frac{\mathbf{v}_{j}}{\sqrt{d_j}}\right)\otimes\left(\mathbf{v}_{i,j}-\mathbf{v}_{j,i}\right)\right)\right.\nonumber\\
 &  & \left.\times\left(\sum_{\mathbf{v}_{i,j}\in\mathcal{E}_{1}}\left(\frac{\mathbf{v}_{i}}{\sqrt{d_i}}-\frac{\mathbf{v}_{j}}{\sqrt{d_j}}\right)\otimes\left(\mathbf{v}_{i,j}-\mathbf{v}_{j,i}\right)\right)^\dagger\right\}\nonumber \\
 & = & \frac{1}{2}\sum_{\mathbf{v}_{i,j}\in\mathcal{E}_{1}}\mathrm{Tr}_{\mathcal{E}}\left\{ \left(\frac{\mathbf{v}_{i}}{\sqrt{d_i}}-\frac{\mathbf{v}_{j}}{\sqrt{d_j}}\right)\left(\frac{\mathbf{v}_{i}}{\sqrt{d_i}}-\frac{\mathbf{v}_{j}}{\sqrt{d_j}}\right)^{\dagger}\otimes\left(\mathbf{v}_{i,j}-\mathbf{v}_{j,i}\right)\left(\mathbf{v}_{i,j}-\mathbf{v}_{j,i}\right)^{\dagger}\right\} \nonumber\\
 & = & \frac{1}{2}\sum_{\mathbf{v}_{i,j}\in\mathcal{E}_{1}}\mathrm{Tr}_{\mathcal{E}}\left\{ \left(\frac{\mathbf{v}_{i}}{\sqrt{d_i}}-\frac{\mathbf{v}_{j}}{\sqrt{d_j}}\right)\left(\frac{\mathbf{v}_{i}}{\sqrt{d_i}}-\frac{\mathbf{v}_{j}}{\sqrt{d_j}}\right)^{\dagger}\otimes\left(\mathbf{v}_{i,j}\mathbf{v}_{i,j}^{\dagger}+\mathbf{v}_{j,i}\mathbf{v}_{j,i}^{\dagger}\right)\right\}\nonumber \\
 & = & \sum_{\mathbf{v}_{i,j}\in\mathcal{E}_{1}}\left(\frac{\mathbf{v}_{i}}{\sqrt{d_i}}-\frac{\mathbf{v}_{j}}{\sqrt{d_j}}\right)\left(\frac{\mathbf{v}_{i}}{\sqrt{d_i}}-\frac{\mathbf{v}_{j}}{\sqrt{d_j}}\right)^{\dagger}= \mathcal{L} ,
\end{eqnarray}
where the final line follows from \eqref{eq:SbarSbarlap}.
\end{proof}
Lemma \ref{lem:SymLaplacianQstate} tells us that we can view
\begin{equation}
\left\vert  \psi_G \right>:= \frac{\psi_G}{\sqrt{\mathrm{Tr} \{\mathcal{L} \}}} \label{eq:psi_G}
\end{equation}
as a purification \cite{wilde2011classical,nielsen2002quantum} of the quantum states 
\begin{equation}
\rho_\mathcal{V} = \frac{\mathcal{L} }{{\mathrm{Tr} \{\mathcal{L} \}}}\label{eq:Gsqustate}
\end{equation} 
and 
\begin{equation}
\rho_\mathcal{E} = \frac{\mathcal{L}^+ }{{\mathrm{Tr} \{\mathcal{L} \}}}.
\end{equation}

We are now in  a position to calculate the Von Neumann entropy of the graph's symmetric Laplacian. This can be done by calculating 
\begin{equation}
H(G) := -\mathrm{Tr}\{\rho_\mathcal{V} \log \rho_\mathcal{V}\}.
\end{equation}
Because $\rho_\mathcal{V} $ and $\rho_\mathcal{E}$ have been obtained from the partial trace over a bipartite pure state (i.e., \eqref{eq:psi_G}), we also have \cite{wilde2011classical}
\begin{equation}
H(G) = -\mathrm{Tr}\{\rho_\mathcal{E} \log \rho_\mathcal{E}\},
\end{equation}
which leads us to the following interpretation for the Von Neumann entropy of a graph's symmetric Laplacian.
\begin{lem}
\label{lem:The-von-Neumann}The Von Neumann entropy of a graph's symmetric Laplacian $\mathcal{L}$ (equivalently  $\mathcal{L}^+$)
(denoted by $H(G) =H\left(\rho_{\mathcal{V}}\right)=H\left(\rho_{\mathcal{E}}\right):=-\mathrm{Tr}\left\{ \rho_{\mathcal{V}}\log\rho_{\mathcal{V}}\right\} $)
is the highest achievable rate that entanglement can be concentrated
from the bipartite pure state $\left|\psi_{G}\right\rangle $, \eqref{eq:psi_G}.
\end{lem}
\begin{proof}
This statement is analogous to \cite[Proposition 4]{de2016interpreting}, which was applied to combinatorial Laplacians instead of symmetric Laplacians. It follows from the concentration theorem \cite[Theorem 19.4.1]{wilde2011classical}.
\end{proof}

\section{Extremal Entropies of Symmetric Graph Laplacians\label{sec:extremal}}

The goal of this section is to determine the extremal values that the symmetric graph Laplacian's entropy can take for connected graphs on $n$ vertices. Before this, we present the R\'enyi-$p$ entropy of a quantum state $\rho$, which is defined to be \cite{renyi1961measures}
\begin{equation}
H_p \left(  \rho   \right) := \frac{1}{1-p}\log\mathrm{Tr} \left\{  \rho^p  \right\}. \label{eq:RenyiEntropydefinition12}
\end{equation}
For a graph $G$ we may use the notation $H_p(G)$ to denote the R\'enyi-$p$ entropy \eqref{eq:RenyiEntropydefinition12} of the quantum state corresponding to $G$, \eqref{eq:Gsqustate}. 
The  R\'enyi-$p$ entropy will be called upon throughout the remainder of this work, and some of its properties are given as follows. 
\begin{properties*}\label{prop:renyi}
The  R\'enyi-$p$ entropy has the following two properties \cite{muller2013quantum}
\begin{description}
\item[\emph{1}] $\lim_{p\to1} H_p \left(  \rho   \right) $ is equal to the Von Neumann entropy of $\rho$.
\item[\emph{2}] $H_p(\rho)$ is a decreasing function of $p$.
\end{description}
\end{properties*}The largest value that the graph's entropy can take is given in the following theorem, and this value is achieved by the complete graph.
\begin{thm} \label{lem:entropycompletegraph}
Let $G$ be a simple connected graph on $n$ vertices and let $\mathcal{K}_n$ be the complete graph on $n$ vertices. Then for $p \geq 1$ we have
\begin{equation}
H_p\left(  G   \right) \leq  H_p\left(  \mathcal{K}_n  \right).
\end{equation}
Also, for $p\to 1$ we have 
\begin{equation}
H\left(  \mathcal{K}_n  \right) = \log (n-1).
\end{equation}
\end{thm}
\begin{proof}
 This follows immediately from the uniformity of the non-zero spectrum for complete graphs (see \cite[Example~1.1]{chung1997spectral}) and Lemma \ref{lem:unifrenyi} below.
\end{proof}

Determining the largest value that the entropy can achieve is somewhat trivial (Theorem \ref{lem:entropycompletegraph}). From our observations, determining the smallest value that the entropy can achieve is much more involved. Indeed, the authors of this paper were unable to establish a statement as strong as that made in Theorem \ref{lem:entropycompletegraph}. Consequently, we present the following theorem, which constructs bounds on the minimum Von Neumann and R\'enyi-$2$ entropy of the symmetric Laplacian.

\begin{thm}\label{thrm:1}Let $G$ be a simple connected graph on $n$ vertices. Then the following statements hold.
\begin{description} 
\item[\emph{A}] We have 
\begin{align}
H_2 \left(  G   \right)  \geq     \log\frac{n^2}{n+\sum_{i}  \frac{1}{\sqrt{d_i}}}  > \log n -\log 2.\label{eq:general2entropy}
\end{align}
Also, 
\begin{align}
H_2 \left(   \mathcal R_{k,n}   \right)  =      \log\frac{n}{1+  \frac{1}{ {k}} },\label{eq:general2entrop22y}
\end{align}
so that the $2$-regular graph minimizes the R\'enyi-$2$ entropy among all $k$-regular graphs.
Consequently,
\begin{align}
H_2 \left(  G   \right)  >   H_2 \left(  \mathcal R_{2,n} \right)  - \log \frac{4}{3}.\label{eq:stuff2q1}
\end{align}
\item[\emph{B}] We have
\begin{equation}
H\left( G\right)   > H\left( \mathcal R_{2,n} \right)  -  \log \frac{4\sqrt{2}}{3} .
\end{equation}
\end{description}
\end{thm}
\begin{proof}
We begin by proving statement $\mathbf{A}$. With 
 $(i\sim j)$ defined according to
\begin{equation}
(i\sim j)=\left\{ \begin{array}{c}
1\;\mathrm{if}\;v_i\;\mathrm{and}\;v_j\;\mathrm{are\;connected}\\
0\;\mathrm{otherwise}
\end{array},\right.\label{eq(ij)}
\end{equation}
we have
\begin{align}
H_2 \left(  G   \right) & = - \log\frac{  \mathrm{Tr} \left\{\left( \Delta ^{-1/2}L \Delta ^{-1/2} \right)^2 \right\} }{n^2}\nonumber\\
& =   - \log\frac{\sum_{i,j} \left[  \Delta ^{-1/2}L \Delta ^{-1/2}  \right]_{i,j}^2}{n^2}\nonumber\\
& =  \log\frac{n^2  }{n+\sum_{i}  \frac{1}{d_i} \left( \sum_{j\neq i}\frac{(i\sim j)}{d_j} \right)}\nonumber\\
& \geq \log\frac{n^2   }{n+\sum_{i}  \frac{1}{\sqrt{d_i} }}\nonumber\\
& > \log\frac{n}{2},
\end{align}
where the first equality follows from \eqref{eq:RenyiEntropydefinition12}, the second equality follows from $\mathrm{Tr}\{ A^2 \} = \sum_{i,j}[A]_{i,j}^2$ for any symmetric matrix $A$, the third equality follows from \eqref{eq(ij)} and \eqref{eq:ijthLaplacian}, and the first inequality follows from Lemma \ref{lem:neighbourinversesum} below. The final inequality follows from $d_i\geq 1$, and noting that we cannot have $d_i=1$ for all $i$. To prove \eqref{eq:general2entrop22y}, note that when $G$ is $k$-regular $\Delta^{-1} L$ is symmetric, so that $\mathrm{Tr}\{ \left(  \Delta^{-1} L\right)^2 \} = \sum_{i,j}[\Delta^{-1} L]_{i,j}^2$, and
\begin{align}
H_2 \left(  \mathcal R_{k,n}    \right) & = - \log\frac{1}{n^2}  \mathrm{Tr} \left\{\left( \Delta ^{-1/2}L \Delta ^{-1/2} \right)^2 \right\} \nonumber\\
& = - \log\frac{1}{n^2}  \mathrm{Tr} \left\{ \left( \Delta ^{-1}L  \right)^2 \right\}\nonumber\\
& = \log \frac{n^2}{n + \sum_i \frac{1}{d_i}},
\end{align}
from which \eqref{eq:general2entrop22y} follows.
Equation \eqref{eq:stuff2q1} follows by rearranging \eqref{eq:general2entrop22y} in terms of $\log n$ and substituting the result into \eqref{eq:general2entropy}.

Statement $\mathbf{B}$ follows from the following set of inequalities
\begin{align}
H\left(  G \right) \geq & H_2(G) \geq  H_2 \left(  \mathcal R_{2,n} \right)  - \log \frac{4}{3}
\geq  H( \mathcal R_{2,n}) - \frac{1}{2} \log 2-    \log \frac{4}{3},
\end{align}
where the first inequality follows from Properties \ref{prop:renyi}, the second inequality follows from statement $\mathbf{A}$ of this theorem, and the final inequality follows from Lemma \ref{lem:h1minush2} below.
\end{proof}

\begin{rem}
Theorem \ref{thrm:1} contrasts sharply with the result that was established in \cite{dairyko2016note} for combinatorial Laplacians. In particular, that work demonstrated that the star graph's combinatorial Laplacian achieves minimum R\'enyi-$2$ entropy and \emph{almost always} achieves minimum Von Neumann entropy.
When considering the symmetric Laplacian, the star graph does not minimize either of these metrics. To see this note that, from \cite[Example 1.2]{chung1997spectral}, the spectrum of the symmetric Laplacian of the complete bipartite graph $\mathcal{K}_{n-k,n}$  is given by $0$, $1$ (with multiplicity $n-2$), and $2$, so that
$$H_2( \mathcal{K}_{n-k,k} ) = - \log \left(\frac{n-2}{n^2} + \left(\frac{2}{n}\right)^2 \right) =  - \log\left( \frac{n+2}{n^2} \right) \geq H_2(\mathcal R_{2,n}).$$
Note, this value is independent of $k$, and the star graph is a special case when $k=n-1$.
Moreover, its Von Neumann entropy is independent of $k$ and given by
\begin{equation}
H  \left(  \mathcal{K}_{n-k,k}   \right) = \log n - \frac{2}{n}\log 2 = \log \frac{n }{2^{2/n}}\sim \log (n-1).\label{eq:stargraphentropysym}
\end{equation}
The limit at the end of \eqref{eq:stargraphentropysym} shows that the star graph actually tends towards achieving maximum entropy (Theorem~\ref{lem:entropycompletegraph}) as $n$ grows large.
\end{rem}

\subsection{Useful Lemmas} 

This subsection contains useful lemmas that are used within the proofs of Theorems~\ref{lem:entropycompletegraph} and~\ref{thrm:1}. Before these, we present the following two definitions.
\begin{defn}
(Majorization \cite{marshall1979inequalities}) Let $a= [a_1, \dots, a_n], b= [b_1, \dots, b_n]\in\mathbb{R}^n$ be $n$ dimensional vectors with elements arranged in non-increasing order (i.e., $a_i\geq a_j$ and $b_i\geq b_j$ for $i<j$). We say that $a$ is majorized by $b$, denoted $a \prec b$, if
\begin{equation}
\sum_{i=1}^{k} a_i \leq \sum_{i=1}^{k} b_i\;\mathrm{for\;all\;}k\quad\mathrm{and}\quad\sum_{i=1}^{n} a_i = \sum_{i=1}^{n} b_i.
\end{equation}
\end{defn}
\begin{defn}\label{def:ShurConc}
(Shur convex (concave) function \cite{roberts1973convex}) Let $f:\mathbb{R}^d\to \mathbb{R}$. We say that $f$ is Shur convex (or Shur concave)  if for $a ,b \in\mathbb{R}^n$ such that $a\prec b$ one has $f(a)\leq f(b) $ (or $f(a)\geq f(b) $).
\end{defn}
\begin{lem}\label{lem:unifrenyi} Let $x = [1/n, 1/n,\dots, 1/n]$
and $y = [y_1, \dots, y_n]$
 be probability vectors on $n$ elements (i.e., $\sum_i x_i = \sum_i y_i = 1$ and $x_i,y_i\geq 0$). Also, without loss of generality suppose that the elements of $x$ and $y$ are  arranged in non-increasing order. Then for $p\geq 1$  
\begin{equation}
\frac{1}{1-p}\log \sum_{i=1}^n x_i^p \geq \frac{1}{1-p}\log \sum_{i=1}^n y_i ^p .
\end{equation}
\end{lem}
\begin{proof}This follows from the  Shur concavity (Definition \ref{def:ShurConc}) of the generalized  R\'enyi-$p$ entropy \cite{ho2015convexity} and Lemma \ref{lem:xmajyunif} below.
\end{proof}
\begin{lem}\label{lem:xmajyunif} Let $x$
and $y$
 be as in Lemma \ref{lem:unifrenyi}.
Then $x\prec y$.
\end{lem}
\begin{proof}
For $1 \leq k < n$, let
$$
S_k = \sum_{i=1}^k y_{i}
\quad
\text{and}
\quad
T_k = \sum_{i=k+1}^n y_{i}.
$$
First, it is easy to see that
\begin{equation}
S_k + T_k = 1 \label{eq:sumSkTk}
\end{equation}
Second, we have
\begin{equation}
S_k \geq \sum_{i=1}^k y_{k} = ky_{k} \Rightarrow \frac{S_k}{k} \geq y_{k}
\quad
\text{and}
\quad
T_k \leq \sum_{i=k+1}^n y_{k} = (n - k)y_{k} \Rightarrow \frac{T_k}{n-k} \leq y_{k},
\end{equation}
so that
\begin{equation}
\frac{n-k}{k}S_k \geq T_k .\label{eq:SkTk}
\end{equation}
Combining \eqref{eq:sumSkTk} and  \eqref{eq:SkTk}, we have
$$
S_k + \frac{n-k}{k}S_k \geq S_k + T_k = 1 \Rightarrow S_k \geq \frac{k}{n} = \sum_{i=1}^k \frac{1}{n}.
$$
Therefore, $x$ is majorized by $y$.
\end{proof}

\begin{lem}\label{lem:h1minush2}
For connected graphs, 
\begin{equation}
0 \leq H (G)  - H_2  (G)  \leq  \frac{1}{2}\log 2.
\end{equation} 
\end{lem}
\begin{proof}
The lower bound follows immediately from the properties of the  R\'enyi-$p$ entropy at the beginning of section~\ref{sec:extremal}. To prove the upper bound, we consider the structural entropy of the graph $H_{str}(G)$, which is defined to be \cite{zyczkowski2003renyi}
\begin{align}
H_{str}(G) :=& H (G)  - H_2 (G) .
\end{align} 
 The structural entropy satisfies the following upper bound $H_{str}(G) \leq   \frac{1}{2} \left(  H_0\left(  G \right) - H_2\left(  G \right)\right)$, \cite[Equation (11)]{zyczkowski2003renyi}. Consequently, we have
\begin{align}
H_{str}(G)  &\leq \frac{1}{2} \left(  H_0\left(  G \right) - H_2\left(  G \right)\right) \\
 &< \frac{1}{2}\left(\log (n-1) - H_2 \left(  \mathcal R_{2,n} \right) + \log \frac{4}{3}  \right)\label{eq:linenumber212} \\
 &= \frac{1}{2}  \left( \log (n-1) - \log n + \log\frac{3}{2} + \log \frac{4}{3}   \right) \\
 &< \frac{1}{2} \left(  \log\frac{3}{2} + \log \frac{4}{3}   \right)\label{eq:linenumber2122}\\
 &= \frac{1}{2} \log 2,
\end{align} 
where the first summand of line \eqref{eq:linenumber212} follows from  the fact that all but one of the graph's spectrum are non-zero, and the final two summands follow from \eqref{eq:stuff2q1}. This gives the result.
\end{proof}

\begin{lem}\label{lem:neighbourinversesum}Let $G$ be a simple connected graph. With 
 $(i\sim j)$ defined according to \eqref{eq(ij)}, we have
\begin{equation}
\frac{1}{d_i}\sum_{j\neq i} \frac{(j\sim i)}{d_j} \leq \frac{1}{\sqrt{d_i}} \label{eq:twoequalities1234}.
\end{equation}
\end{lem}
\begin{proof}
We have
\begin{align}
\frac{1}{d_i}\sum_{j\neq i} \frac{(j\sim i)}{d_j} & \leq \frac{1}{\sqrt{d_i}} \sum_{j\neq i} \frac{(j\sim i)}{\sqrt{d_id_j}} =  \frac{1}{\sqrt{d_i}} ,
\end{align}
where the first inequality follows because $d_i\geq 1$ for all $i$ and the first equality follows from the basic properties of the symmetric Laplacian (i.e., the sum over any of its rows is $0$ and the elements of the diagonal are all $1$).
\end{proof}

\section{Conclusions\label{sec:conc}}

In this work, we offer a rigorous interpretation of the symmetric Laplacian as the partial trace of a bipartite pure quantum state that lives in a vertex space and an edge space. The implication  of this is that the Von Neumann entropy of the symmetric Laplacian can be interpreted as a measure of bipartite entanglement present within this bipartite pure state. We also present results on the extreme values of the entropy for simple connected graphs. The lower extremal values contrast sharply with similar results applied to combinatorial Laplacians. Future work should be performed to establish the topological implications of graph Von Neumann entropy.

 \section*{Acknowledgments}
The authors wish to acknowledge the support of EPSRC under grant number (EP/K04057X/2) and the UK National Quantum Technologies Programme under grant number (EP/M013243/1).

\bibliographystyle{unsrt}
\bibliography{animeshrefs}

\end{document}